\documentclass[11pt]{article}
\usepackage[utf8]{inputenc}
\usepackage{amsmath,amssymb,bbm}
\usepackage{algorithm}
\usepackage{algorithmicx}
\usepackage{algpseudocode}
\usepackage[whole]{bxcjkjatype}
\usepackage{graphicx}[dvipdfmx]
\usepackage{mdframed}
\usepackage{natbib}
\usepackage{hyperref}
\hypersetup{
    colorlinks = true,
	citecolor=blue,
	linkcolor=blue
}
\usepackage{geometry}
\usepackage{fourier}
\usepackage{authblk}
\usepackage{enumerate}
\usepackage{amsthm}
\usepackage{subcaption}
\theoremstyle{definition}

\newtheorem*{theorem*}{Theorem}

\newtheorem{prop}{Proposition}
\newtheorem*{note*}{Note}

\newtheorem{ex}{Example}

\newcommand{\argmin}{\mathop{\arg\min}}

\newcommand{\diff}{\mathrm{d}}

\newcommand{\toprob}{\mathop{\to}^\text{in prob.}}
\newcommand{\ig}{(\text{IG})}
\newcommand{\gompertz}{(\text{Gom})}
\newcommand{\gm}{(\text{GM})}

\usepackage{multirow}
\newcommand{\revisebegin}{}
\newcommand{\reviseend}{}

\title{Minimizing robust density power-based divergences for general parametric density models}
\author[1,2]{Akifumi Okuno}
\affil[1]{The Institute of Statistical Mathematics}
\affil[2]{RIKEN Center for Advanced Intelligence Project}
\date{\empty}

\begin{document}

\maketitle

\begin{abstract}
Density power divergence (DPD) is designed to robustly estimate the underlying distribution of observations, in the presence of outliers. 
However, DPD involves an integral of the power of the parametric density models to be estimated; the explicit form of the integral term can be derived only for specific densities, such as normal and exponential densities. 
While we may perform a numerical integration for each iteration of the optimization algorithms, the computational complexity has hindered the practical application of DPD-based estimation to more general parametric densities. 
To address the issue, this study introduces a stochastic approach to minimize DPD for general parametric density models. The proposed approach can also be employed to minimize other density power-based $\gamma$-divergences, by leveraging unnormalized models. 
\revisebegin 
We provide \verb|R| package for implementation of the proposed approach in \url{https://github.com/oknakfm/sgdpd}.
\reviseend
\end{abstract}


\section{Introduction}

As the presence of outliers within observations may adversely affect the statistical inference, robust statistics has been developed for several decades~\citep{huber1981robust,hampel1986robust,maronna2006robust}. 
Amongst many possible directions, the divergence-based approach, which estimates some parameters in probabilistic models by minimizing the divergence to underlying distributions, has drawn considerable attention owing to its compatibility with the probabilistically formulated problems. 
In particular, density power divergence (DPD), also known as $\beta$-divergence \citep{basu1998robust}, extends the Kullback-Leibler divergence to enhance robustness against outliers. DPD has gained recognition as one of the most widely used divergences across disciplines. 
DPD finds applications in various fields, including blind source separation~\citep{minami2002robust}, matrix factorization~\citep{tan2012automatic}, signal processing~\citep{basseville2013divergence}, Bayesian inference~\citep{ghosh2016robust}, variational inference~\citep{futami2018variational}, and more, contributing to the enhancement of robustness in these applications.

DPD comprises an integral term of the power of the parametric density models to be estimated. 
Unfortunately, however, an explicit form of this integral term can be obtained only for specific density functions including normal density~\citep{basu1998robust}, exponential density~\citep{jones2001comparison}, generalized Pareto density~\citep{juarez2004robust}, Weibull density~\citep{basu2016generalized}, and 
generalized exponential density~\citep{hazra2022minimum}; most of the existing literature considers only the normal density function. The computational difficulty of this integral term has hindered the application of DPD-based estimation to general parametric density models, over a quarter of a century since the proposal of DPD. 
To overcome this computational limitation, for instance, \citet{pmlr-v89-okuno19b} considers matching mean functions instead of matching probability densities, \citet{fujisawa2006robust} minimizes an upper-bound of DPD for normal mixture, and \citet{kawashima2019robust} and \citet{nandy2022robust} compute finite approximations of the intractable term for Poisson and skew-normal distributions, respectively. 
While these challenging attempts may offer individual solutions for specific parametric density estimations, they cannot be extended to the countless other types of parametric density models, including inverse-normal and Gompertz densities (refer to, for instance, \citet{krishnamoorthy2006handbook} for a comprehensive list of statistical distributions). 
Such a stringent limitation on the choice of parametric density models can lead to unwanted model misspecification, even in cases where the form of the underlying density function is already roughly specified. 
Hence, there is a considerable demand for a general approach to minimize DPD for a wide range of parametric density models.


One straightforward method to deal with the integral term is to perform numerical integration. However, in the context of parameter estimation problems, employing gradient descent (GD) necessitates the computation of numerical integration at each iteration, resulting in a non-negligible computational complexity. 
Furthermore, numerical integration may introduce a non-negligible approximation error. Therefore, in this study, we introduce a simple stochastic optimization algorithm~\citep{robbins1951stochastic,ghadimi2013stochastic}, which addresses both of these issues simultaneously. An illustration of this approach, including the estimation of Gompertz and normal mixture density models, can be found in Figure~\ref{fig:illustration}. Notably, the proposed optimization approach can also be viewed as a robust divergence modification of the contrastive divergence learning~\citep{hinton2002training,perpinan2005contrastive}.

This study also explores the minimization of another well-known density power-based $\gamma$-divergence~\citep{fujisawa2008robust} using the same approach, leveraging the unnormalized models~\citep{kanamori2015robust}.

\begin{figure}[!ht]
\centering 
\begin{minipage}{0.49\textwidth}
\centering
\includegraphics[width=\textwidth]{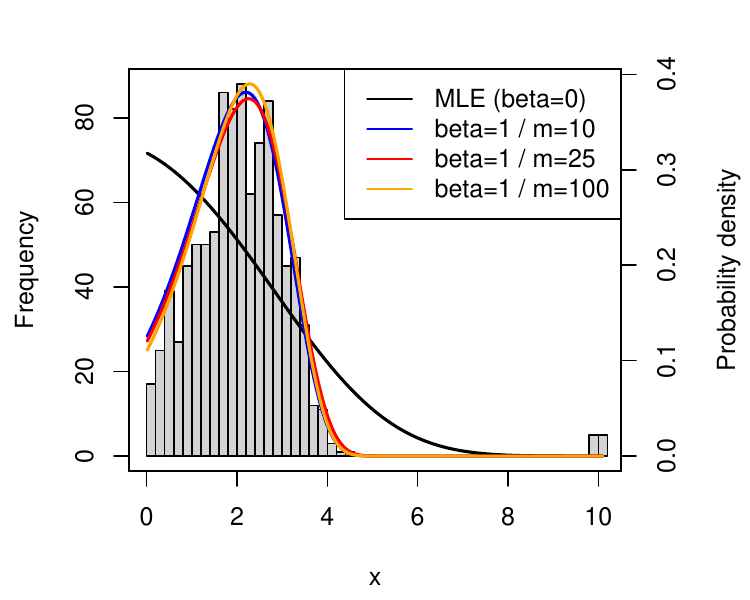}
\subcaption{Gompertz distribution.}
\label{fig:hist_gompertz}
\end{minipage}
\begin{minipage}{0.49\textwidth}
\centering
\includegraphics[width=\textwidth]{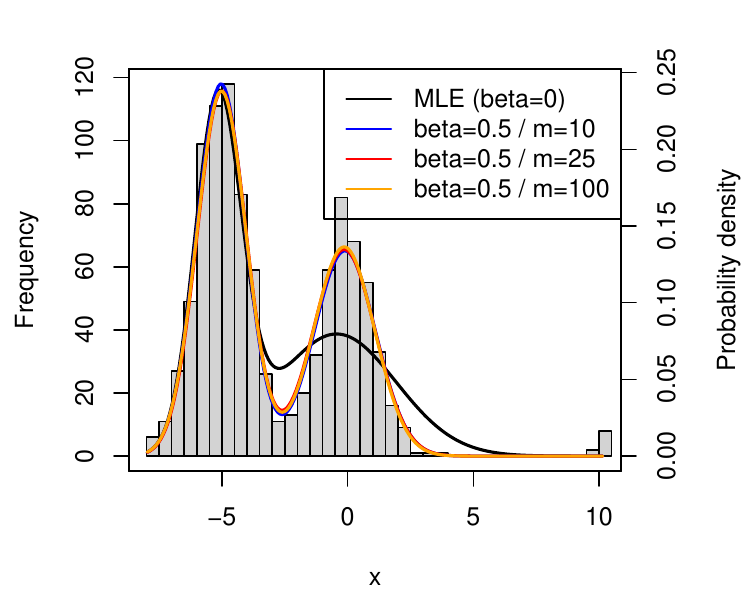}
\subcaption{Normal mixture distribution.}
\label{fig:hist_GM}
\end{minipage}
\caption{Gompertz and normal mixture models, where the explicit forms of integral terms of the powered densities cannot be obtained, are estimated by the proposed stochastic optimization algorithm under the outlier contamination. The black line represents the maximum likelihood estimator, and the colored lines represent the density power estimators. In these experiments, outliers are generated from the normal distribution $N(10,1)$ with the contamination ratio $\xi=0.01$. 
See Section~\ref{sec:experiments} for more details. 
}
\label{fig:illustration}
\end{figure}

\section{Density power estimator}

Background and problem setting are presented in Sections~\ref{subsec:background} and \ref{subsec:problem}, respectively. Our proposed solution, a stochastic gradient descent (SGD) algorithm utilizing an unbiased stochastic gradient, is detailed in Section~\ref{subsec:proposal}. 
Additionally, Section~\ref{subsec:comparison} offers a discussion comparing our approach with GD using numerical integration.

\subsection{Background: density power divergence and estimator}
\label{subsec:background}
Let $d \in \mathbb{N}$ and let $\mathcal{X} \subset \mathbb{R}^d$. 
Suppose that given vectors $x_1,x_2,\ldots,x_n \in \mathcal{X}$ are independently and identically drawn from a distribution $Q$ whose support is $\mathcal{X} \subset \mathbb{R}^d$. 
This study considers estimating the underlying distribution $Q$ by a parametric distribution $P_{\theta}$, using the observed vectors $x_1,x_2,\ldots,x_n$. 
$p_{\theta}$ denotes a probability density function of the distribution $P_{\theta}$. 

Let $\beta>0$ be a user-specified hyperparameter, typically specified as $\beta=0.1,0.5$ or $\beta=1$. 
$\beta$ is also called a power-parameter. 
The \emph{density power divergence}~(DPD, also known as $\beta$-divergence; \citet{basu1998robust}, \citet{eguchi2001robustifing}, \citet{basu2011statistical}) between the underlying distribution $Q$ and the parametric distribution $P_{\theta}$ (equipped with the parameter $\theta \in \Theta \subset \mathbb{R}^s,s \in \mathbb{N}$) is defined by 
\begin{align*}
    D_{\beta}(Q,P_{\theta}):=d_{\beta}(Q,P_{\theta})-d_{\beta}(Q,Q),
\end{align*}
where 
\begin{align}
    d_{\beta}(Q,P_{\theta})
    =
    -\frac{1}{\beta} \int_{\mathcal{X}} p_{\theta}(x)^{\beta} \diff Q(x) 
    +
    r_{\theta}^{(\beta)},
    \quad 
    \left(
    r_{\theta}^{(\beta)}
    =
    \frac{1}{1+\beta} \int_{\mathcal{X}} p_{\theta}(x)^{1+\beta} \diff x
    \right)
    \label{eq:DPCE}
\end{align}
is referred to as \emph{density power cross entropy~(DPCE)}. 
While the above DPD and DPCE are defined for the continuous distribution $Q$, the integral should be replaced with the discrete summation if $Q$ is a discrete distribution. 
DPD can be regarded as a discrepancy measure between two distributions $Q,P_{\theta}$; we can estimate the parameter $\theta$ in the model $P_{\theta}$ by minimizing the DPD. The DPD reduces to the well-known Kullback-Leibler divergence 
$D(Q,P_{\theta})=d(Q,P_{\theta})-d(Q,Q)$, $d(Q,P_{\theta})=\int \log p_{\theta}(x) \diff Q(x)$ by taking the limit $\beta \searrow 0$.

DPD has a robustness property. 
Consider the case that $Q$ is composed of the true distribution $P_{\theta_*}$ and the outlier distribution $R$, i.e., 
\begin{align}
Q = (1-\xi) P_{\theta_*} + \xi R
\label{eq:contaminated_distribution}
\end{align}
with the contamination ratio $\xi > 0$. 
As $R$ represents the outlier distribution, we may assume that $\nu^{(\beta)}(\theta) := \frac{1}{\beta}\int p_{\theta}(x)^{\beta} \diff R(x) \ge 0$ is small enough with positive power-parameter $\beta>0$ and $\theta \approx \theta_*$. 
Then, it holds for $\theta \approx \theta_*$ that
\begin{align*}
    d_{\beta}(Q,P_{\theta})
    &=
    d_{\beta}((1-\xi)P_{\theta_*}+\xi R,P_{\theta})
    =
    d_{\beta}((1-\xi)P_{\theta_*},P_{\theta})
    -
    \nu^{(\beta)}(\theta) \\
    &\approx 
    d_{\beta}((1-\xi)P_{\theta_*},P_{\theta});
\end{align*}
the divergence $d_{\beta}(Q,P_{\theta})$ mitigates the negative impact of the outlier distribution $R$. Consequently, we anticipate that the estimator based on DPD will be in closer proximity to the true underlying parameter $\theta_*$ when compared to the maximum likelihood estimator, which corresponds to the case $\beta=0$.

While the explicit form of the underlying distribution $Q$ is not obtained in practice, the DPCE~\eqref{eq:DPCE} is empirically approximated by substituting the empirical distribution $\hat{Q}(x)=\frac{1}{n}\sum_{i=1}^{n} \mathbbm{1}(x_i \le x)$ into the distribution $Q$. 
The empirical DPCE is defined by 
\begin{align}
    d_{\beta}(\hat{Q},P_{\theta})
    =
    - \frac{1}{\beta} \frac{1}{n} \sum_{i=1}^{n} p_{\theta}(x_i)^{\beta} 
    +
    r_{\theta}^{(\beta)},
    \label{eq:empirical_DPCE}
\end{align}
and the empirical density power~(DP) estimator $\hat{\theta}_{\beta}$ is defined as 
\begin{align}
    \hat{\theta}_{\beta} = \argmin_{\theta \in \Theta} D_{\beta}(\hat{Q},P_{\theta})
    =
    \argmin_{\theta \in \Theta}d_{\beta}(\hat{Q},P_{\theta}).
    \label{eq:DPestimator}
\end{align}
The empirical DP estimator reduces to the maximum likelihood estimator by taking the limit $\beta \searrow 0$.

The definitions and properties of DPD, as mentioned earlier, are extensively covered in existing literature. For more detailed information, see, e.g., \citet{basu1998robust}, \citet{jones2001comparison}, \citet{cichocki2010families}, \citet{basu2011statistical}, and \citet{kanamori2015robust}.

\subsection{Problem: computational difficulty}
\label{subsec:problem}

Unfortunately, computing the DP estimator~\eqref{eq:DPestimator} becomes intractable for general parametric density models $p_{\theta}$ due to the unavailability of an explicit form for the integral term $r_{\theta}^{(\beta)}$. Explicit forms of the integral term $r_{\theta}^{(\beta)}$ are obtained only for certain specific parametric density models, such as the normal density and exponential density. For instance, if $p_{\theta}$ represents the density function of the univariate normal distribution $N(\mu,\sigma^2)$, we can obtain $r_{\theta}^{(\beta)}=(2\pi\sigma^2)^{-\beta/2}(1+\beta)^{-3/2}$. The challenge lies in determining the extent to which we can obtain explicit forms for the integral term $r_{\theta}^{(\beta)}$.

While this extent has not been fully clarified to the best of the authors' knowledge, it is worth noting that the we can calculate the integral term $r_{\theta}^{(\beta)}$ for exponential family distributions whose density functions take the form: 
\begin{align}
    p_{\theta}(x) = h \exp\left( \langle \theta, u(x)\rangle - v(\theta) \right).
    \label{eq:exponential_family}
\end{align}
$h \ge 0, u:\mathcal{X} \to \mathbb{R}^s$ and $v:\Theta \to \mathbb{R}$ are user-specified parameter and functions, respectively. The family of these densities includes normal, exponential, and Gamma distributions, and we obtain the explicit form of $r_{\theta}^{(\beta)}$ as shown in Proposition~\ref{prop:r_exponential}. 
The proof is straightforward as shown in Appendix~\ref{app:proof_prop:r_exponential}.

\begin{prop}
\label{prop:r_exponential}
    $r_{\theta}^{(\beta)}=\frac{1}{1+\beta}h^{\beta}\exp\left( v((1+\beta) \theta)-(1+\beta) v(\theta) \right)$ for the parametric density \eqref{eq:exponential_family}. 
\end{prop}

While there exist several remaining distributions whose explicit form of $r_{\theta}^{(\beta)}$ can be obtained (see, e.g., \citet{juarez2004robust} for generalized Pareto distribution, and \citet{basu2016generalized} for Weibull distribution), $r_{\theta}^{(\beta)}$ is in general still computationally intractable. 
Following the proof shown in Appendix~\ref{app:proof_prop:r_exponential}, we can easily observe that Proposition~\ref{prop:r_exponential} cannot be extended to the exponential family with non-constant base measure $h(x)$, i.e., 
$p_{\theta}(x)=h(x) \exp(\langle \theta,u(x)\rangle-v(\theta))$. 
This family includes inverse-normal distribution, Poisson distribution, and so forth. One possible approach to compute $r_{\theta}^{(\beta)}$ for such density functions is leveraging the numerical integration. 
For instance, if the support of the general parametric density model $p_{\theta}(x)$ is $\mathbb{R}$, we may compute a numerical integration
\begin{align}
    \widehat{r}_{\theta}^{\beta}
    :=
    \frac{1}{1+\beta}
    \frac{2D}{M}
    \sum_{j=1}^{M} p_{\theta}(z_j)^{1+\beta}, 
    \label{eq:naive_numerical_int}
\end{align}
with $z_j=-D+2D(j-1)/(M-1)$. 
In the limit as $D$ and $M$ tend to infinity, it is expected that \eqref{eq:naive_numerical_int} converges to $r_{\theta}^{(\beta)}$. Therefore, we can utilize gradient-based optimization algorithms, using the gradient approximated by the numerical integration \eqref{eq:naive_numerical_int}. 
For the discrete Poisson case, \citet{kawashima2019robust} employs a similar idea for parameter estimation, although they use a slightly different but essentially equivalent divergence called $\gamma$-divergence~\citep{fujisawa2008robust}. 
For the skew normal case, \citet{nandy2022robust} also computes the finite approximation.

However, approximating $r_{\theta}^{(\beta)}$ with high accuracy necessitates choosing large values for $D$ and $M$. 
To make matters worse, minimizing the DPD in this way requires computing the gradient of the numerical integration~\eqref{eq:naive_numerical_int} for each iteration, resulting in non-negligibly high computational complexity.

\subsection{Proposal: a stochastic approach to minimize DPD}
\label{subsec:proposal}

To dodge the high computational complexity described in Section~\ref{subsec:problem}, this study employs a SGD using an unbiased stochastic gradient.

We first describe the idea in an intuitive manner. 
While the approach described in the last of Section~\ref{subsec:problem} intended to compute the ``exact'' integration for optimization, this study considers employing a ``rough'' estimate of the gradient, i.e., a stochastic unbiased estimator of the exact gradient. 
With the stochastic gradient $g(\theta^{(t)}; \zeta^{(t)})$ satisfying 
\begin{align}
\mathbb{E}_{\zeta^{(t)}}(g(\theta^{(t)}; \zeta^{(t)}))=\frac{\partial}{\partial \theta}d_{\beta}(\hat{Q},P_{\theta^{(t)}})
\label{eq:unbiasedness}
\end{align}
(and $\|\mathbb{V}_{\zeta^{(t)}}(g(\theta^{(t)};\zeta^{(t)}))\|<\infty$), whose randomness is specified by the random variable $\zeta^{(t)}$, conventional theories prove that the parameter $\theta^{(t)}$ iteratively updated by SGD
\begin{align}
    \theta^{(t)} = \theta^{(t-1)} - \eta_t g(\theta^{(t-1)}; \zeta^{(t-1)})
    \quad (t=1,2,\ldots,T)
    \label{eq:sgd}
\end{align}
with the decreasing learning rate $\eta_t \searrow 0$ yields the convergence of $\frac{\partial}{\partial \theta}d_{\beta}(\hat{Q},P_{\theta^{(t)}})$ to $0$ (as $t \to \infty$). See Proposition~\ref{prop:convergence} for a more rigorous description. 
Note that the learning rate $\eta_t$ should be decreased to $0$ in stochastic optimization, while the full-batch GD usually considers a constant learning rate.

The stochastic gradient $g(\theta^{(t)}; \zeta^{(t)})$ can incorporate various forms of randomness. One example of the stochastic gradient $g(\theta^{(t)};\zeta^{(t)})$ is $\frac{\partial}{\partial \theta}d_{\beta}(\hat{Q},P_{\theta^{(t)}})+\zeta^{(t)}$, where $\zeta^{(t)}$ is a standard normal random number. Another example is $g(\theta^{(t)};\zeta^{(t)})=2 \zeta^{(t)} \frac{\partial}{\partial \theta}d_{\beta}(\hat{Q},P_{\theta^{(t)}})$ with the random number $\zeta^{(t)}$ assuming values of $1$ and $0$ with probabilities of $1/2$, respectively. In both cases, these constructions satisfy the unbiasedness conditions outlined in \eqref{eq:unbiasedness}, ensuring that SGD effectively minimizes DPD. 
A key point to note here is the flexibility in designing stochastic gradients; we may employ various stochastic gradient, tailored to specific computational or theoretical requirements. As will be discussed later in this section, this study specifically introduces a stochastic gradient formulation based on a finite summation, meticulously defined for efficient computation. 

Traditionally speaking, this type of optimization algorithm has its roots in \citet{robbins1951stochastic}, and a similar idea can be found in maximum likelihood estimation of computationally intractable probability models~\citep{geyer1992constrained,hinton2002training,perpinan2005contrastive}. 
A significant benefit to employing such a stochastic optimization algorithm is that we can dodge the exact computation of the gradient, which is computationally intensive due to the numerical integration~\eqref{eq:naive_numerical_int}.

\bigskip

In DPD-based estimation, we can define an unbiased stochastic gradient $g(\theta^{(t)};\zeta^{(t)})$ as follows. 
Let $\tilde{P}_t$ be a user-specified \revisebegin proposal \reviseend distribution, such that random numbers can be generated from $\tilde{P}_t$. $\tilde{p}_t$ denotes its probability density function. 
\revisebegin 
As also mentioned later in Example~\ref{ex:gaussian_mixture}, we generally assume that 
the proposal distribution $\tilde{p}_t$ has the same support as $p_{\theta^{(t)}}$, and $\tilde{p}_t$ should be square integrable. 
\reviseend 
With functions $w_t(y)=p_{\theta^{(t)}}(y)/\tilde{p}_t(y)$ \revisebegin (where we assume $0/0=0$) \reviseend and $t_{\theta}(x)=\partial \log p_{\theta}(x)/\partial \theta$, we randomly generate independent $m \in \mathbb{N}$ samples $\zeta^{(t)}=(y_1^{(t)},y_2^{(t)},\ldots,y_m^{(t)})$ from $\tilde{P}_t$, and define
\begin{align}
    g(\theta^{(t)};\zeta^{(t)})
    &=
    -\frac{1}{n}\sum_{i=1}^{n}p_{\theta^{(t)}}(x_i)^{\beta}t_{\theta^{(t)}}(x_i) \nonumber \\
    &\hspace{7em}+
    \frac{1}{m}\sum_{j=1}^{m}
    w_t(y_j^{(t)})
    p_{\theta^{(t)}}(y_j^{(t)})^{\beta}t_{\theta^{(t)}}(y_j^{(t)}).
    \label{eq:sfo_gradient}
\end{align}

While we employ $\tilde{P}_t=P_{\theta^{(t)}}$ in our experiments (so as to obtain constant weight function $w_t(y)=1$) for simplicity, we may employ normal distribution or some other computationally-tractable distributions to generate $\zeta^{(t)}$. 

Interestingly, \eqref{eq:sfo_gradient} obviously satisfies the unbiasedness assumption~\eqref{eq:unbiasedness} (by taking the expectation with respect to $\zeta^{(t)}$), \emph{regardless of the sample size $m \in \mathbb{N}$}. 
Even if we employ $m=1$, the SGD equipped with \eqref{eq:sfo_gradient} is proved to minimize DPD (though the optimization procedure can be slightly unstable when $m \in \mathbb{N}$ is excessively small); this approach much reduces the computational complexity, as it does not need to take the limit $m \to \infty$ unlike the numerical integration~\eqref{eq:naive_numerical_int}. 
Our numerical experiments in Section~\ref{sec:experiments} demonstrate that $m=10$ is enough for plausible computation. Further discussions on the comparison with GD using numerical integration can be found in Section~\ref{subsec:comparison}. 

A slightly simpler version of the convergence theorem shown in \citet{ghadimi2013stochastic} Theorem 2.1(a) is summarized in Proposition~\ref{prop:convergence}. 
While this study considers only the simple stochastic algorithm, several options such as variance reduction~\citep{wang2013variance} and their theories can also be incorporated into our approach.

\begin{prop}
\label{prop:convergence} 
Let $m \in \mathbb{N},L,v>0$ and let $\Theta=\mathbb{R}^s$. 
Let $\|\theta\|=\{\theta_1^2+\theta_2^2+\cdots+\theta_s^2\}^{1/2}$ be $2$-norm. 
Assume that 
\begin{enumerate}[{(i)}]
\item $f(\theta)=d_{\beta}(\hat{Q},P_{\theta})$ is differentiable over $\Theta$, 
\item $\|\partial f(\theta)/\partial \theta-\partial f(\theta')/\partial \theta\| \le L\|\theta-\theta'\|$ for any $\theta,\theta' \in \Theta$, 
\item $\mathbb{E}_{\zeta^{(t)}}(g(\theta^{(t)};\zeta^{(t)}))=\partial f(\theta^{(t)})/\partial \theta$, and 
\item $\mathbb{E}_{\zeta^{(t)}}(\|g(\theta^{(t)};\zeta^{(t)})-\partial f(\theta^{(t)})/\partial \theta\|^2) \le v$,  
\end{enumerate}
for any $\theta^{(t)} \in \Theta$ and for any $t \in \{1,2,\ldots,T\}$. 
Also assume that the learning rate $\{\eta_t\}_{t=1}^{T}$ satisfies 
$\eta_t < 2/L$, $\sum_{t=1}^{T} \eta_t \to \infty$ and $\{\sum_{t=1}^{T}\eta_t\}^{-1}\sum_{t=1}^{T}\eta_t^2 \to 0$ as $T \to \infty$. 
Then, the sequence $\{\theta^{(t)}\}$ obtained by the SGD~\eqref{eq:sgd} satisfies 
\[
    \mathbb{E}_{\tau}\left( 
    \bigg\|
    \frac{\partial}{\partial \theta}d_{\beta}(\hat{Q},P_{\theta^{(\tau)}}) 
    \bigg\|^2
    \right)
    \toprob
    0 
    \quad 
    (T \to \infty). 
\]
$\mathbb{E}_{\tau}$ represents the expectation with respect to the number of iterations $\tau \in \{1,2,\ldots,T\}$ randomly chosen with the probability $\mathbb{P}(\tau=k \mid T)=\{2\eta_k-L\eta_k^2\}/\sum_{k=1}^{T}\{2\eta_k-L\eta_k^2\}$ ($k=1,2,\ldots,T$). 
\end{prop}

The number of iterations, denoted as $\tau$, is randomly determined according to the probability $\mathbb{P}(\tau=k \mid T)$. This randomness is crucial for addressing the challenges posed by the non-convexity of the function $f(\theta)$ to be minimized. 

considering $\lim_{T \to \infty} \mathbb{P}(\tau \le T' \revisebegin \mid T \reviseend)=0$ for any fixed $T' \in \mathbb{N}$, it becomes apparent that a larger value of $\tau$ is more likely to be selected as $T$ increases.
Additionally, given the constraint $\{\sum_{t=1}^{T} \eta_t\}^{-1}\sum_{t=1}^{T}\eta_t^2 \to 0$, which necessitates a decreasing learning rate $\eta_t$, the expectation is effectively taken over a smaller region around the local minima in the parameter space $\Theta$, as the learning rates $\eta_t$ are small at sufficiently large number of iterations $t \in \mathbb{N}$. 
Therefore, this theorem suggests that at sufficiently large iteration numbers, the estimated parameter forces the gradient of $d_{\beta}(\hat{Q},P_{\theta^{(t)}})$ to approximate $0$.

While assumptions (i)--(iv) in Proposition~\ref{prop:convergence} may appear restrictive, especially in the outer regions of the parameter space $\Theta$, it is possible to tailor the functions (only in the outer region, which is typically less relevant for parameter estimation) to meet the conditions (i) to (iv). These conditions (i)--(iv) are essential for mathematical rigor and completeness. 
An example of densities satisfying these conditions (i)--(iv) is shown in Example~\ref{ex:gaussian_mixture}. 

\begin{ex}\label{ex:gaussian_mixture}
Let $\phi(x;\mu,\sigma^2)$ represents the univariate normal density with mean $\mu \in \mathbb{R}$ and variance $\sigma^2>0$. Let $\rho(z):=(1+\exp(-z))^{-1}$ be the sigmoid function and let $\varepsilon>0$ be a sufficiently small constant. 
Let $p_{\theta}(x)=\rho(\theta_1)\phi(x;\theta_2,\theta_3^2+\varepsilon)+\{1-\rho(\theta_1)\} \phi(x;\theta_4,\theta_5^2+\varepsilon)$ be a normal mixture density model defined with the parameter vector $\theta=(\theta_1,\theta_2,\theta_3,\theta_4,\theta_5) \in \Theta = \mathbb{R}^5$. 
\revisebegin 
We assume the following (C-1)--(C-3) on the proposal density $\tilde{p}_t$: 
(C-1) it has the same support as $p_{\theta^{(t)}}$, 
(C-2) it is square-integrable, and 
(C-3) $\mathbb{E}(w_t(y)^4)<\infty$ for $y \sim \tilde{p}_t$, where $w_t(y)=p_{\theta^{(t)}}(y)/\tilde{p}_t(y)$. 
(For instanfce, $\tilde{p}_t=p_{\theta^{(t)}}$ satisfy (C-1)--(C-3).) 
\reviseend 
Then, (i)--(iv) in Proposition~\ref{prop:convergence} hold. 
\end{ex}

In Appendix~\ref{app:proof_of_Example_ex:gaussian_mixture}, the normal mixture density defined in Example~\ref{ex:gaussian_mixture} is proved to satisfy the assumptions (i)--(iv). 
In Example \ref{ex:gaussian_mixture}, the mixture rate is estimated via the sigmoid function $\rho(\theta_1)$ so that $\theta_1$ can take values over the entire real line $\mathbb{R}$, and the variance of each normal density is lower-bounded by $\varepsilon>0$. While this $\varepsilon>0$ is needed to ensure that the normal density has non-zero variance, this restriction can be practically ignored by specifying small $\varepsilon>0$. 
\revisebegin 
Also, by ignoring the mixture rate and the second density and following the same proof, we can easily prove that the non-mixtured, single normal model $p_{\theta}(x)=\phi(x;\theta_1,\theta_2^2+\varepsilon)$ also satisfies (i)--(iv) in Propositon~\ref{prop:convergence}. 
\reviseend
We have also verified the practical convergence of SGD through numerical experiments, as detailed in Section~\ref{sec:experiments}.

\bigskip
A key point of the proposed approach is that the stochastic gradient is used to approximate the gradient of the integral term $r_{\theta}^{(\beta)}$ but not the finite summation $-(\beta n)^{-1}\sum_{i=1}^{n}p_{\theta}(x_i)^{\beta}$ discussed in \citet{kawashima2019robust}. 
To the best of the authors' knowledge, none of the previous studies have considered this simple but effective approach.


Furthermore, in addition to the simple optimization problem discussed above, it is worth noting a challenging attempt in recent research by \citet{yonekura2023adaptation}. They focus on Bayesian inference, particularly computing the posterior distribution, while not solely aiming at estimating a single-point estimator $\hat{\theta}_{\beta}$. Their primary objective is to select the power-parameter $\beta$ using the approach presented in \citet{jewson2022ngeneral}. In their work, the power-parameter $\beta$ (rather than $\theta$) is updated using GD, with the gradient being stochastically approximated through sequential Monte Carlo samplers. Although their approach has the potential to be applied with general parametric models, it addresses a slightly different Bayesian problem. 
Their numerical experiments exclusively employ normal densities, and they generate a substantial number of Monte Carlo samples ($m=2000$) to approximate the gradient in each iteration. In contrast, our experiments consider smaller sample sizes for $m$ (even $m=3$), yet still yield meaningful results with theoretical guarantee.

\subsection{Discussion: a comparison with numerical integration-based approaches}

\label{subsec:comparison}

This section discusses the comparison with GD using numerical integration~\eqref{eq:naive_numerical_int}. 
The advantages of the stochastic optimization approach far outweigh the disadvantages as follows. 

\begin{enumerate}[{(i)}]
\item \textbf{Learning rate}. A key distinction between SGD and GD lies in their typical learning rates. SGD needs to employ a decreasing learning rate, which serves to mitigate the randomness inherent in the stochastic gradient, thereby facilitating convergence. 
In contrast, GD generally employs a constant learning rate, as the full-batch gradient naturally converges to $0$. There is no need for the learning rate in GD to reduce to 0.

\item \textbf{Approximation error}. 
As the number of iterations $T$ approaches infinity, stochastic optimization estimators converge to the precise (local) minima. However, GD employing numerical integration \eqref{eq:naive_numerical_int} with a typical constant learning rate incurs a finite approximation error of order $O(1/\sqrt{M})$, regardless of the iteration count $T$ reaching infinity. 
Nonetheless, by adopting a decreasing learning rate and implementing numerical integration using random numbers drawn independently in each iteration, GD can be regarded as a variation of SGD with a large $m \in \mathbb{N}$. Therefore, in this modified approach, GD also attains the exact estimator as $T$ approaches infinity.

\item \textbf{Computational efficiency}. The stochastic optimization approaches need to generate $m \in \mathbb{N}$ samples, and convergence is proven regardless of the value of $m$. On the other hand, GD using numerical integration generates $M$ samples, and $M$ should be large to reduce the approximation error; its computational complexity is then significantly increased compared to stochastic approaches. 

\item \textbf{Ease of implementation}. As stochastic optimization is extensively employed to optimize deep neural networks in recent years~\citep{Goodfellow-et-al-2016}, many of practical packages are provided. For instance, we may employ \verb|PyTorch|\footnote{\url{https://pytorch.org/}} for implementation.

\item \textbf{Compatibility with non-convex optimization problems}. Robust density power-based divergences are generally known to be non-convex functions with respect to the model parameter $\theta \in \Theta$. As a result, optimization algorithms may get trapped in local minima. However, recent studies have shown that stochastic algorithms have the potential to escape from local solutions due to the increased randomness in each iteration. See, e.g., \citet{jin2017how} and \citet{jin2021nonconvex} that provide more insights into this phenomenon. 
Although it is mainly studied to elucidate the success of deep learning, this property is also useful for robust estimation.  
\end{enumerate}

For more rigorous analyses, also see \citet{nemirovski2009robust} that theoretically proves for certain problem classes that the stochastic optimization significantly outperforms the approach based on numerical integration. 
In Section~\ref{subsec:exp_comparison}, we also compare the proposed stochastic approach with the GD using numerical integration.

\section[Application to the minimization of gamma-divergence]{Application to the minimization of $\gamma$-divergence}
\label{sec:gamma}

More recently, the $\gamma$-divergence~\citep{fujisawa2008robust,kanamori2015robust} $D_{\gamma}(Q,P_{\theta})=d_{\gamma}(Q,P_{\theta})-d_{\gamma}(Q,Q)$ defined with the $\gamma$-cross entropy~(GCE)
\begin{align}
    d_{\gamma}(Q,P_{\theta})
    =
    -\frac{1}{\gamma} \log \int_{\mathcal{X}} p_{\theta}(x)^{\gamma} \diff Q(x)
    +
    \frac{1}{1+\gamma} \log \int_{\mathcal{X}} p_{\theta}(x)^{1+\gamma} \diff x
    \label{eq:GCE}
\end{align}
has attracted considerable attention~\citep{chen2013robust,dawid2016minimum,futami2018variational,castilla2022robust,li2022robust}. 
$\gamma$-divergence is equivalent to a pseudo-spherical score~\citep{good1971comment}. 
$\gamma$-divergence has similar robust properties as the DPD, and it also comprises the integral of the powered density, which is in general computationally intractable. 
While several optimization approaches for $\gamma$-divergence including the Majorize-Minimization algorithm~\citep{hunter2004tutorial,hirose2017robust} have been developed, the parametric density models are still limited to normal density or several specific ones discussed so far.

The proposed approach cannot be directly employed for the optimization of GCE~\eqref{eq:GCE} since applying the $\log$ function to the integral term introduces bias in the stochastic gradient. 
Nevertheless, our approach can still be utilized to minimize the GCE with the aid of unnormalized models, by following the discussion in \citet{kanamori2015robust}. 
An unnormalized model is defined as a general nonnegative function $f:\mathcal{X} \to \mathbb{R}_{\ge 0}$, while the probability density function should satisfy the integral constraint $\int f(x) \diff x=1$. \citet{kanamori2015robust} provides an important identity between minimizers of GCE and DPCE:
\[
    \hat{\theta}_{\gamma}
    =
    \argmin_{\theta \in \Theta}d_{\gamma}(\hat{Q},P_{\theta})
    =
    \argmin_{\theta \in \Theta}\left\{ 
        \min_{c > 0} d_{\beta}(\hat{Q}, cP_{\theta})
    \right\} \bigg|_{\beta = \gamma}. 
\]
$c>0$ is called a scale parameter to be estimated. Using this identity, we can minimize GCE as follows.

Let $\zeta^{(t)}=(y_1^{(t)},\ldots,y_m^{(t)})$ be i.i.d. generated from the distribution $\tilde{P}_t$, where $\tilde{P}_t,t_{\theta}$ and the weight function $w_t$ are defined in Section~\ref{subsec:proposal}. 
Consider a SGD for the augmented parameter $\psi=(\theta,c)$: 
\begin{align}
    \psi^{(t)} = \psi^{(t-1)} - \eta_t \tilde{g}(\psi^{(t-1)}; \zeta^{(t-1)}),
\end{align}
where 
\begin{align}
    \tilde{g}(\psi^{(t)};\zeta^{(t)})
    &=
    (
    \tilde{g}^{(\theta)}(\tilde{\psi}^{(t)};\zeta^{(t)}),
    \tilde{g}^{(c)}(\tilde{\psi}^{(t)};\zeta^{(t)})
    ), 
    \label{eq:sfo_gradient_augmented}
    \\
    \tilde{g}^{(\theta)}(\psi^{(t)};\zeta^{(t)})
    &=
    -
    (c^{(t)})^{\gamma}
    \frac{1}{n}\sum_{i=1}^{n} p_{\theta^{(t)}}(x_i)^{\gamma}t_{\theta^{(t)}}(x_i) \nonumber \\
    &\hspace{7em}+
    (c^{(t)})^{1+\gamma}
    \frac{1}{m}\sum_{j=1}^{m}
    w_t(y_j^{(t)}) p_{\theta^{(t)}}(y_j^{(t)})^{\gamma}t_{\theta^{(t)}}(y_j^{(t)}), \nonumber \\
    \tilde{g}^{(c)}(\psi^{(t)};\zeta^{(t)})
    &=
    -(c^{(t)})^{\gamma-1}\frac{1}{n}\sum_{i=1}^{n}p_{\theta^{(t)}}(x_i)^{\gamma} \nonumber \\
    &\hspace{7em}+
    (c^{(t)})^{\gamma}
    \frac{1}{m}\sum_{j=1}^{m}w_t(y_j^{(t)})p_{\theta^{(t)}}(y_j^{(t)})^{\gamma}. \nonumber
\end{align}
As the stochastic gradient \eqref{eq:sfo_gradient_augmented} is an unbiased estimator of the gradient, i.e., $\mathbb{E}_{\zeta^{(t)}}(\tilde{g}(\psi^{(t)};\zeta^{(t)})) =  \frac{\partial}{\partial \psi}d_{\beta}(\hat{Q},c^{(t)} p_{\theta^{(t)}})$, 
the SGD equipped with \eqref{eq:sfo_gradient_augmented} is expected to produce the scale estimator $\hat{c}_{\gamma}$ and the $\gamma$-estimator $\hat{\theta}_{\gamma}$, which minimizes the $\gamma$-divergence.


\section{Numerical experiments}
\label{sec:experiments}

This section describes the numerical experiments. 
Note that \verb|R| source codes to reproduce the experimental results are provided in \url{https://github.com/oknakfm/DPD}. 
\revisebegin 
Although the implementation details slightly differ from those used in our numerical experiments, we also offer an R package for the proposed approach, available at \url{https://github.com/oknakfm/sgdpd}.
\reviseend

\subsection{Robust parameter estimation}
\label{subsec:estimation_intractable}

To demonstrate the proposed approach, we synthetically generate $n=1000$ observations $x_1,x_2,\ldots,x_n$ from the contaminated distribution~\eqref{eq:contaminated_distribution}:
\[
    Q=(1-\xi)P_{\theta_*}+\xi R. 
\]
Particularly, we employ for (the outlier distribution) $R$ the normal distribution with mean $\mu=10$ and the standard deviation $\sigma=1$. 
For the parametric distribution $P_{\theta}$, we empoly the following four types:
\begin{enumerate}
    \item \textbf{Normal distribution}, with the true parameter 
    \[
    \theta_*=(\mu_*,\sigma_*)=(0,1).
    \]
    \item \textbf{Inverse normal distribution}, with the true parameter 
    \[
    \theta_*=(\mu_*,\lambda_*)=(1,3).
    \]
    \item \textbf{Gompertz distribution}, with the true parameter 
    \[
    \theta_*=(\omega_*,\lambda_*)=(1,0.1).
    \]
    \item \textbf{Normal mixture distribution}, with the true parameter 
    \[
    \theta_*=(\mu_{1*},\sigma_{1*},\mu_{2*},\sigma_{2*},\alpha_*)=(-5,1,0,1,0.6).
    \]
\end{enumerate}
See Appendix~\ref{app:distribution} for the definitions of the distributions (ii)--(iv). 
We compute the density power estimators by conducting the SGD with the stochastic gradient~\eqref{eq:sfo_gradient}. The parameters to be estimated are initialized by the maximum likelihood estimator. 
Learning rate $\eta_t$ is decreased by multiplying the decay rate $r=0.7$ for each $25$ iterations: 
remaining settings $(T,\eta_0)$ of the SGD are (i) $(500,1)$, (ii) $(1000,1)$, (iii) $(1000,0.5)$, (iv) $(1000,1)$. 
Contamination ratios are $\xi=0.1$ for (i) and (ii), and $\xi=0.01$ for (iii) and (iv). 
Estimation results of (i)--(iv) are shown in Figures~\ref{fig:hist_normal}, \ref{fig:hist_IG}, \ref{fig:illustration}(\subref{fig:hist_gompertz}) and \ref{fig:illustration}(\subref{fig:hist_GM}), respectively. 

\bigskip

Firstly, for the setting (i) normal distribution, we can compute the exact empirical DPCE by virtue of Proposition~\ref{prop:r_exponential} (while the integral terms cannot be computed for the remaining settings (ii)--(iv)). 
We monitor the empirical DPCE via the SGD iterations. See Figure~\ref{fig:empiricalDPCE}. We can observe that the SGD significantly reduces the DPCE. 
All the DPCE get close to almost the same value regardless of the sample size $m \in \mathbb{N}$, while the DPCE sequence has a larger variance if smaller $m$ is employed.

\begin{figure}[!ht]
\centering
\includegraphics[width=0.3\textwidth]{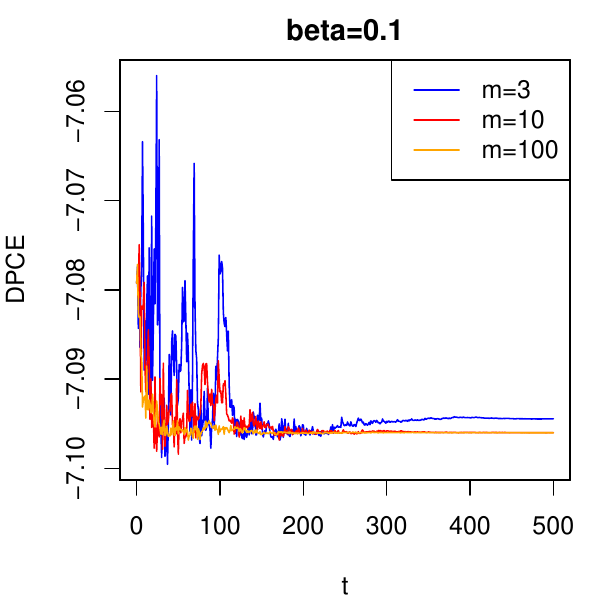}
\includegraphics[width=0.3\textwidth]{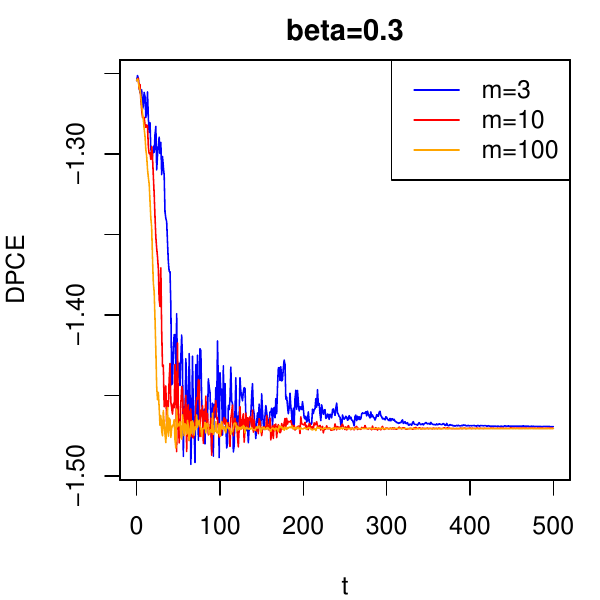}
\includegraphics[width=0.3\textwidth]{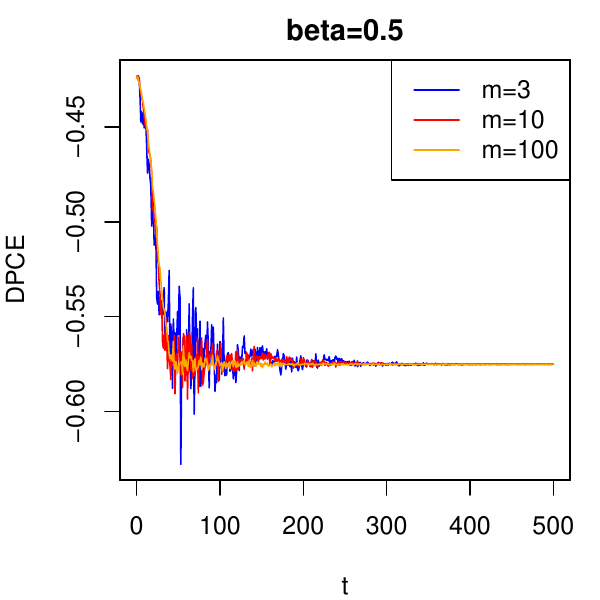}
\caption{Convergence of the (empirical) DPCE when $P_{\theta}$ is a normal distribution. 
\label{fig:empiricalDPCE}
}
\end{figure}

Also see Figure~\ref{fig:hist_normal}; the DPD-based estimators fit to the true normal distribution robustly against the outlier distribution $R=N(10,1)$. The estimated results are not much different for different $m \in \mathbb{N}$, and larger $\beta>0$ provides more robust estimators. 
For the settings (ii)--(iv), see Figures~\ref{fig:hist_IG}, \ref{fig:illustration}(\subref{fig:hist_gompertz}), and \ref{fig:illustration}(\subref{fig:hist_GM}). 
Almost the same tendency can be observed.

\begin{figure}[!ht]
\centering 
\begin{minipage}{0.47\textwidth}
\includegraphics[width=\textwidth]{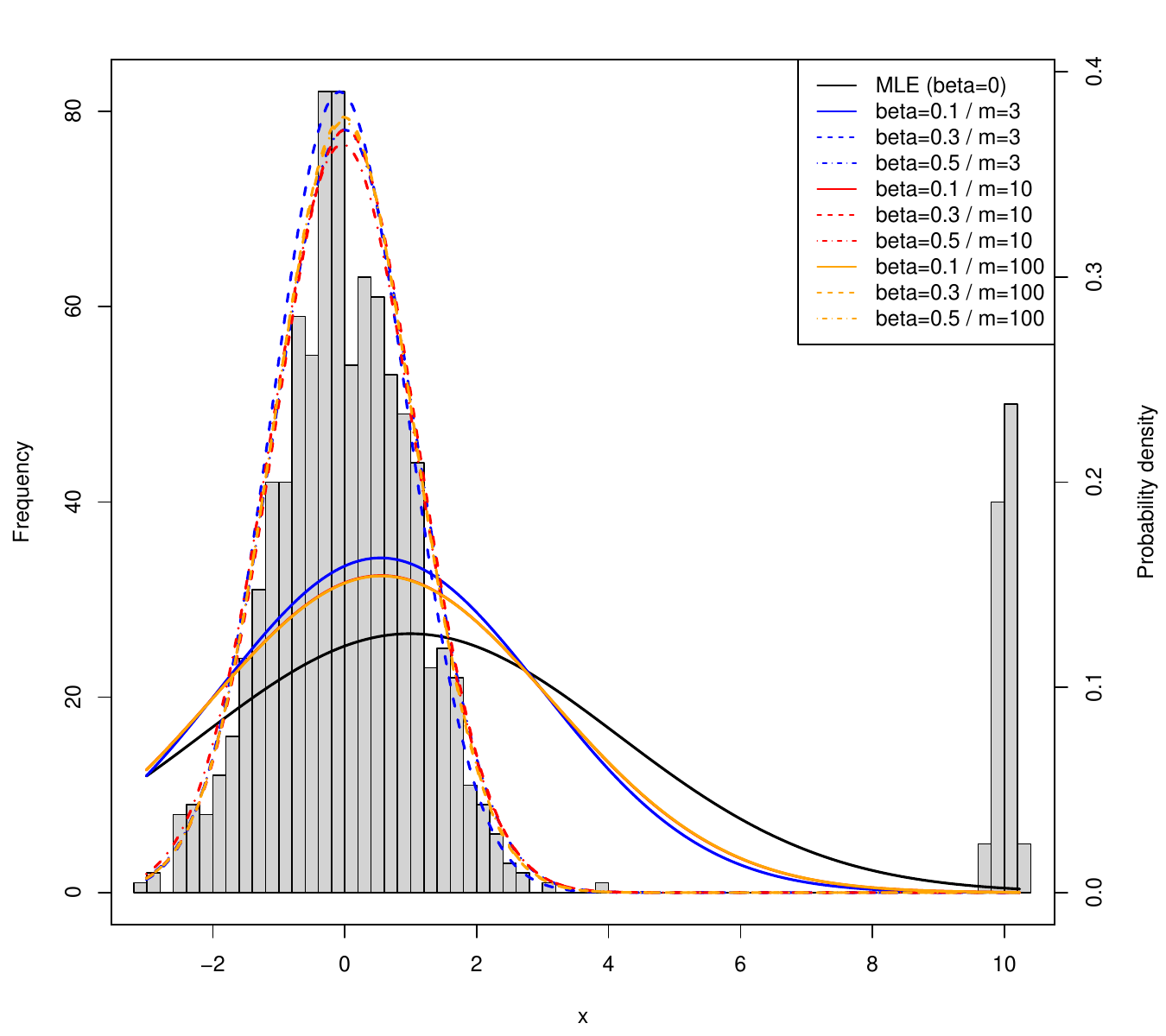}
\caption{\textbf{Normal distribution} with the true parameter $(\mu_*,\sigma_*)=(0,1)$.}
\label{fig:hist_normal}
\end{minipage}
\hspace{1em}
\begin{minipage}{0.47\textwidth}
\includegraphics[width=\textwidth]{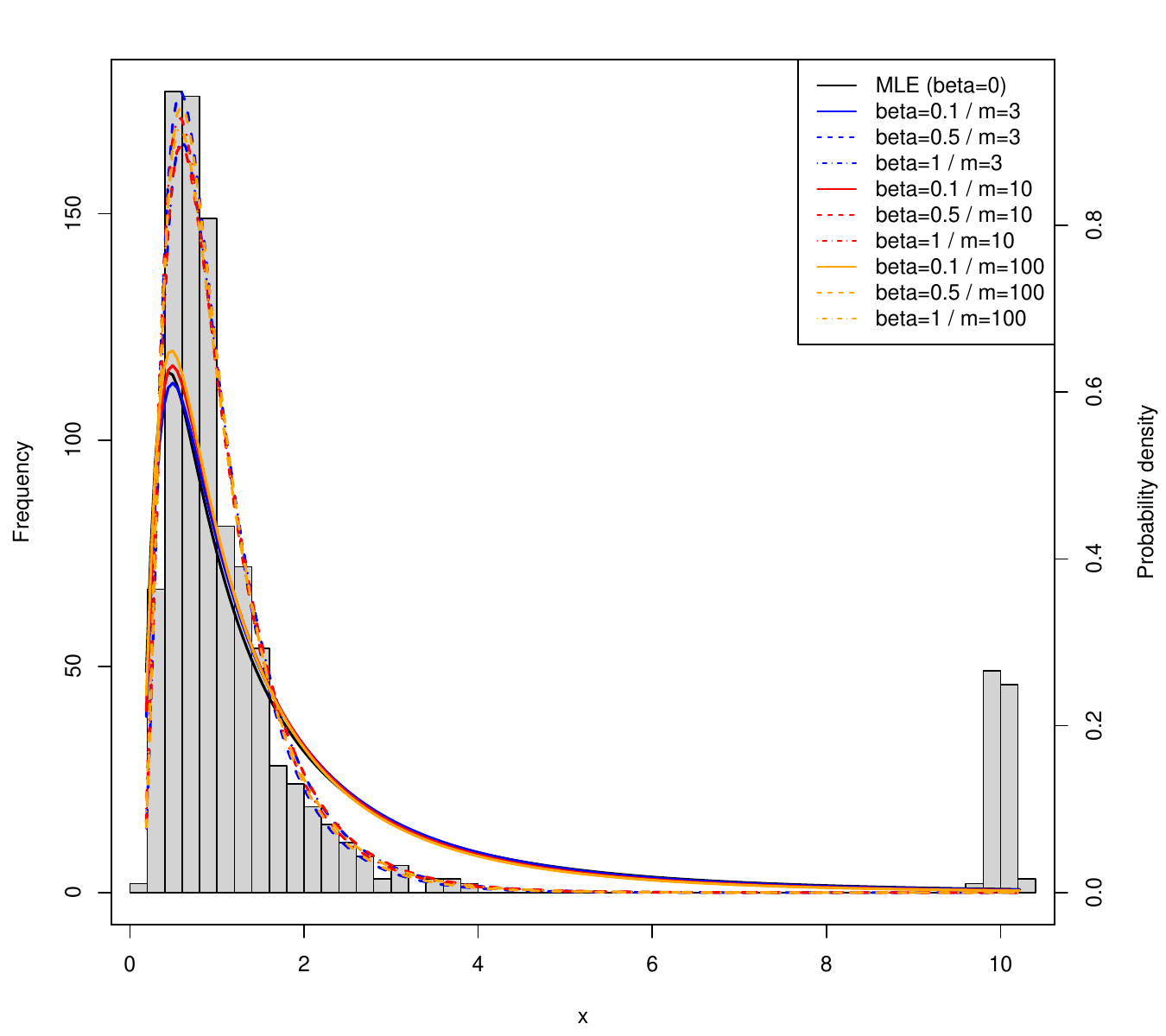}
\caption{\textbf{Inverse normal dist.} with the true parameter $(\mu_*,\lambda_*)=(1,3)$.}
\label{fig:hist_IG}
\end{minipage}
\end{figure}


Using the same setup as in (i) with the normal distribution, we also minimize the $\gamma$-cross entropy (GCE). The settings remain consistent, and the initial value of the scale parameter $c>0$ is set to $c=1$. We computed the exact values of empirical GCE, as depicted in Figure~\ref{fig:empiricalGCE}. It is noticeable that the empirical GCE decreases similarly to the empirical DPCE. The resulting distributions closely resemble those of the DPD while the figures are not included in this text due to the space limitation. 
Throughout the iterations, we also monitor the scale parameter $\hat{c}>0$ for unnormalized models, as discussed in Section~\ref{sec:gamma}. This is illustrated in Figure~\ref{fig:scaling}. It is apparent that the scale parameter converges to the true value $1-\xi=0.9$.

\begin{figure}[!ht]
\centering
\includegraphics[width=0.3\textwidth]{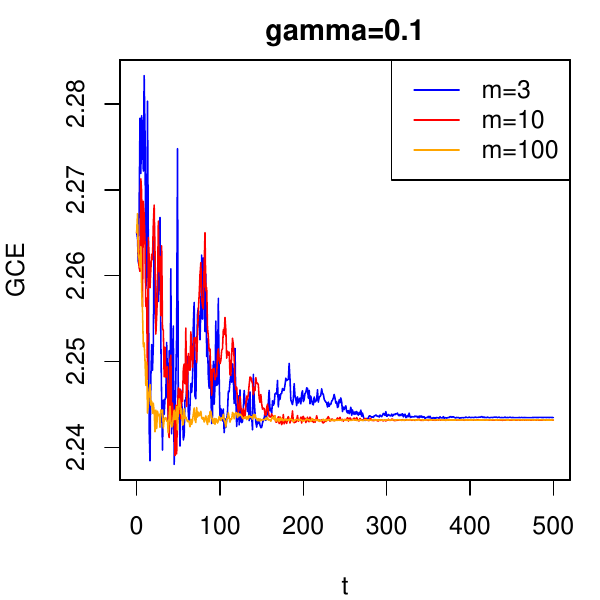}
\includegraphics[width=0.3\textwidth]{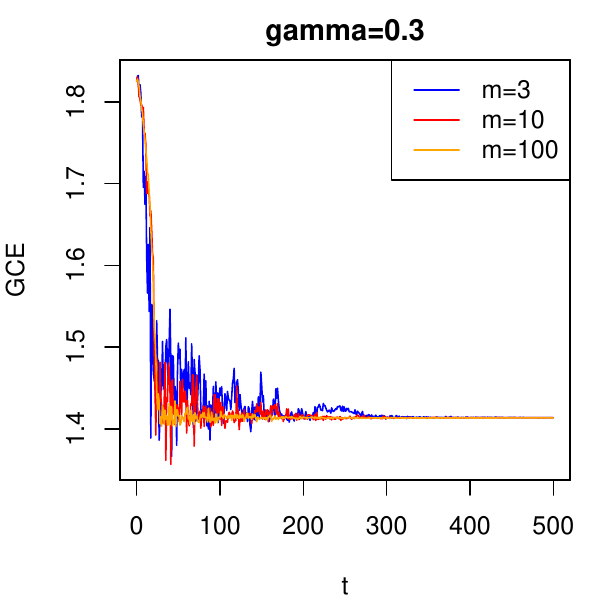}
\includegraphics[width=0.3\textwidth]{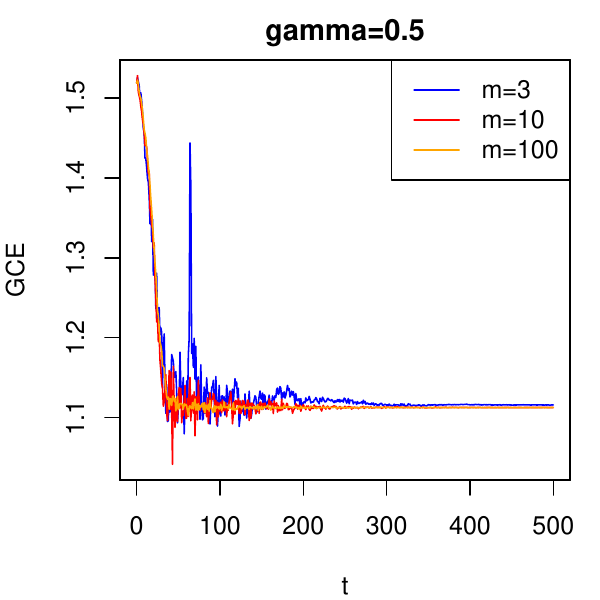}
\caption{Convergence of the (empirical) GCE when $P_{\theta}$ is a normal distribution.}
\label{fig:empiricalGCE}
\end{figure}

\begin{figure}[!ht]
\centering
\includegraphics[width=0.3\textwidth]{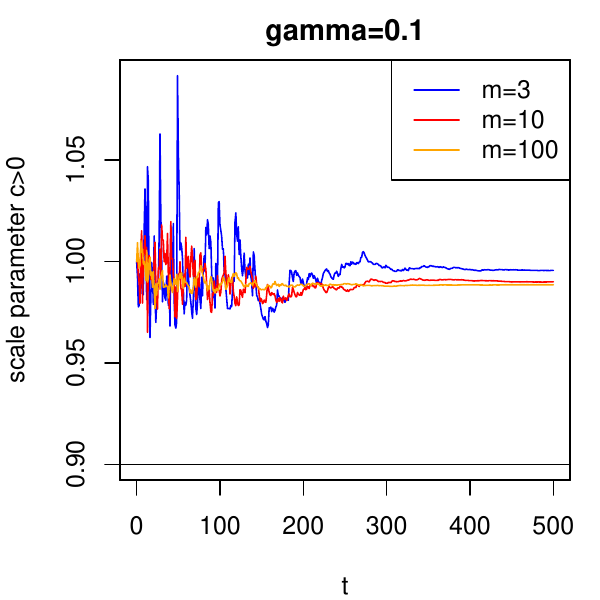}
\includegraphics[width=0.3\textwidth]{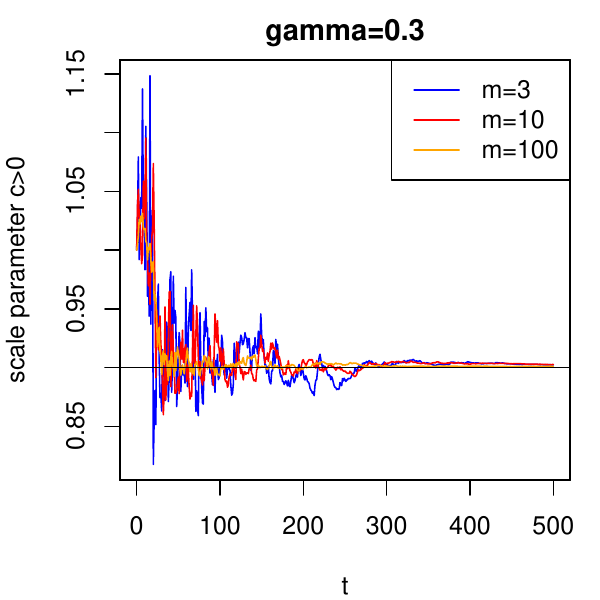}
\includegraphics[width=0.3\textwidth]{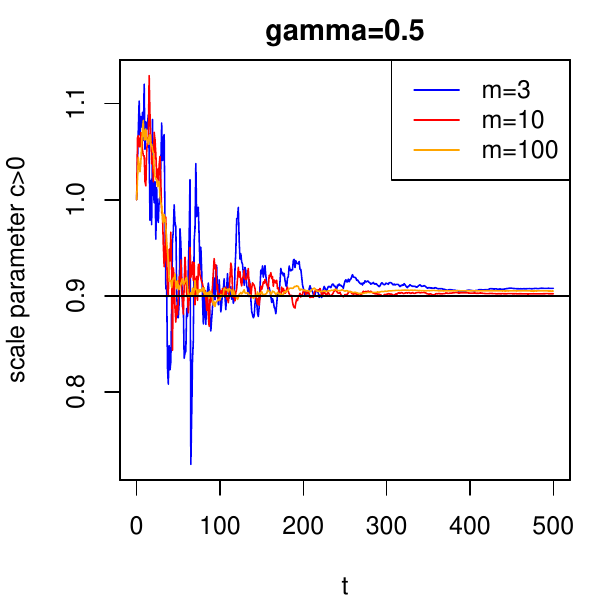}
\caption{Scale parameter $\hat{c}$ obtained by minimizing empirical GCE. It is preferable that $\hat{c}$ converges to $1-\xi=0.9$ (black line). 
}
\label{fig:scaling}
\end{figure}

\subsection{Comparison with GD using numerical integration}
\label{subsec:exp_comparison}

This section provides a comparison with the GD using numerical integration. For the experiments, we consider a $d$-variate normal distribution 
\begin{align}
p_{\theta}(x) = (2\pi)^{-d/2}\exp(-\|x-\theta\|_2^2/2)
\label{eq:normal_model}
\end{align}
with the true parameter $\theta_*=0.5 \mathbbm{1}_d$, where $\mathbbm{1}_d$ represents a $d$-dimensional vector whose entries are all $1$. Outlier distribution is also a normal distribution with mean $\theta_* + 100 \mathbbm{1}_d$ and the variance covariance matrix $0.01 I_d$. We synthetically generate $n=500$ observations with the contamination ratio $\xi=0.01$. 
Namely, this dataset contains $5$ outliers. 
Using the contaminated dataset, we estimate the parameter $\theta$ in the normal density model \eqref{eq:normal_model}. We consider the following two methods for optimization, which are initialized by the maximum likelihood estimator. The number of iterations is $T=300$, and the power-parameter is $\beta=0.5$. 

\begin{itemize}
\item \textbf{Stochastic approach}: we utilize the same setup as described in Section~\ref{subsec:estimation_intractable}. The learning rate $\eta_t$, which initially starts at $\eta_0=1$, is decreased by a factor of $0.7$ after every $20$ iteration.

\item \textbf{Baseline}: We calculate GD using numerical integration. Numerical integration is carried out on a regular grid comprising $M=3^d,10^d,50^d$ lattice points within the region $[-2,2]^d$. The constant learning rate $\omega$ is defined as $T^{-1}\sum_{t=1}^{T}\eta_t$.
\end{itemize}

We calculate the error $\|\hat{\theta}^{(t)}-\theta_*\|_2^2$ for each estimate $\hat{\theta}^{(t)}$ using each method. To ensure a fair comparison, we graph the error alongside the complexity $t(n+m)$ (for SGD) and $t(n+M)$ (for GD using numerical integration) for iterations $t=1,2,\ldots,T$. These complexities represent the cumulative number of observations and samples used to compute the (stochastic or numerically integrated) gradient.

The results are presented in Figure~\ref{fig:comparison_NI}.
As shown in Figure~\ref{fig:comparison_NI}(\subref{fig:d=2}), GD using numerical integration effectively minimizes the error. However, for higher dimensions, such as $d=3$ as seen in Figure~\ref{fig:comparison_NI}(\subref{fig:d=3}), GD becomes unstable, unlike the stochastic approach. This instability becomes more pronounced with increasing dimension $d$. For example, parameter estimation via numerical integration diverges when $d=4$. 
While we can potentially stabilize the computation by
(i) increasing the number of lattice points $M$ for numerical integration,
(ii) using a smaller learning rate, 
(iii) increasing the number of iterations $T$, 
all of these solutions lead to increased computational complexity. Therefore, we conclude that the stochastic approach is much more efficient than GD using numerical integration.

\begin{figure}[ht]
\begin{minipage}{0.48\textwidth}
\includegraphics[width=\textwidth]{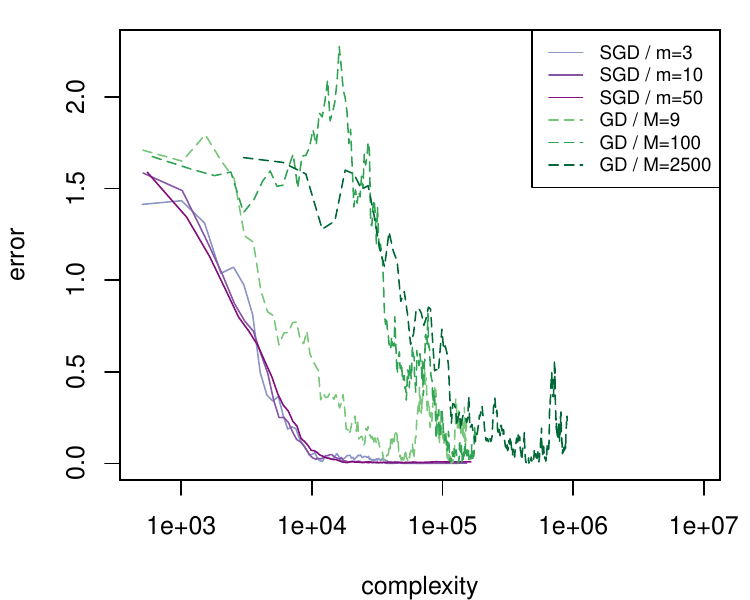}
\subcaption{$d=2$}
\label{fig:d=2}
\end{minipage}
\begin{minipage}{0.48\textwidth}
\includegraphics[width=\textwidth]{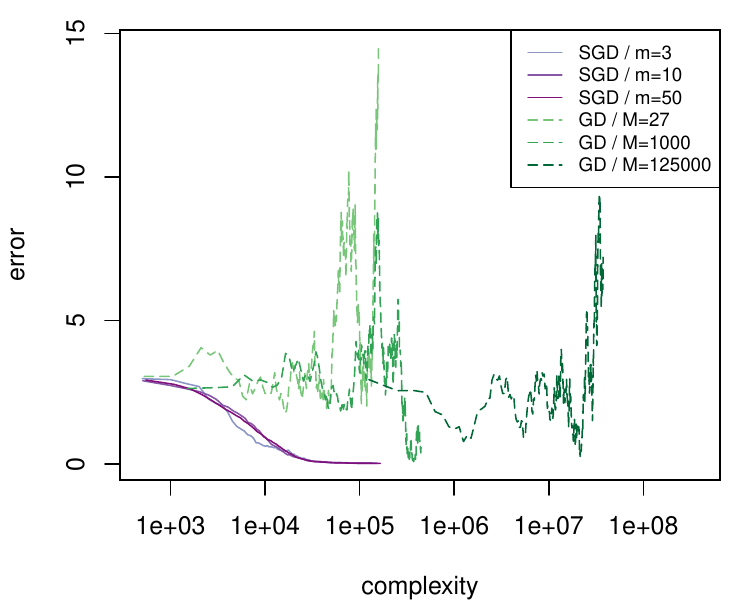}
\subcaption{$d=3$}
\label{fig:d=3}
\end{minipage}
\caption{Comparison with GD using numerical integration.}
\label{fig:comparison_NI}
\end{figure}

To rigorously validate the aforementioned observations, we further conduct similar experiments $10$ times. For each setting, we calculate the averaged MSE at the last iteration of both SGD and GD using numerical integration. The results of these experiments are detailed in Tables~\ref{table:d=2} and \ref{table:d=3}, and these results also indicate that SGD effectively minimizes DPD, compared with the GD using numerical integration.

\begin{table}[!ht]
\centering
\caption{Averaged MSE for $10$ times experiments (at the last iteration), for $d=2$. The standard deviation is shown in parentheses. The baseline is the GD using numerical integration (NI).}
\label{table:d=2}
\begin{tabular}{llcc}  
& & MSE & Complexity \\ 
\hline
\multirow{3}{*}{SGD} & $m=3$ & $0.0041 \, (0.0031)$ & $150900$ \\
& $m=10$ & $0.0023 \, (0.0014)$ & $153000$ \\
& $m=50$ & $0.0011 \, (0.0006)$ & $165000$ \\
\hline
\multirow{3}{*}{GD+NI} & $M=3^2$ & $0.1275 \, (0.1159)$ & $152700$ \\
& $M=10^2$ & $0.1462 \, (0.2379)$ & $180000$ \\
& $M=50^2$ & $0.1440 \, (0.1095)$ & $900000$ \\
\hline
\end{tabular}
\end{table}

\begin{table}[!ht]
\centering
\caption{Averaged MSE for $10$ times experiments  (at the last iteration), for $d=3$. The standard deviation is shown in parentheses. The baseline is the GD using numerical integration (NI).}
\label{table:d=3}
\begin{tabular}{llcc}  
& & MSE & Complexity \\ 
\hline
\multirow{3}{*}{SGD} & $m=3$ & $0.0099 \, (0.0059)$ & $150900$ \\
& $m=10$ & $0.0103 \, (0.0039)$ & $153000$ \\
& $m=50$ & $0.0089 \, (0.0016)$ & $165000$ \\
\hline
\multirow{3}{*}{GD+NI} & $M=3^3$ & $10.228 \, (16.716)$ & $158100$ \\
& $M=10^2$ & $4.1233 \, (4.0781)$ & $450000$ \\
& $M=50^2$ & $4.4707 \, (5.1095)$ & $3765000$ \\
\hline
\end{tabular}
\end{table}

\section{Conclusion}
This study provided a stochastic optimization approach to minimize the robust density power divergence for general parametric density models. 
This study explained its adequacy with the aid of conventional theories on stochastic gradient descent. 
Our stochastic approach was also used to minimize $\gamma$-divergence with the aid of unnormalized models. 
Numerical experiments were conducted to validate the proposed approach. 
We also provided \verb|R| package for the proposed approach in \url{https://github.com/oknakfm/sgdpd}.

\section*{Acknowledgement}
A. Okuno was supported by JSPS KAKENHI (21K17718, 22H05106). 
We would like to thank Atsushi Nitanda, Hironori Fujisawa, Shintaro Hashimoto, Takayuki Kawashima, Kazuharu Harada, Keisuke Yano, and Takahiro Kawashima for the helpful comments.

\appendix

\section{Proof of Proposition~\ref{prop:r_exponential}}
\label{app:proof_prop:r_exponential}

As we have 
\begin{align*}
p_{\theta}(x)^{1+\beta}
&=
h^{1+\beta} \exp\left( (1+\beta) \langle \theta,u(x)\rangle -(1+\beta) v(\theta) \right) \\
&=
h^{1+\beta} \exp \left(\langle (1+\beta)\theta,u(x)\rangle - 
v((1+\beta)\theta)
+
v((1+\beta)\theta)
-
(1+\beta) v(\theta)
\right) \\
&=
h^{\beta} \exp \left( v((1+\beta)\theta)
-
(1+\beta) v(\theta) \right)
\underbrace{
h \exp\left(
    \langle (1+\beta)\theta,u(x)\rangle - 
v((1+\beta)\theta)
\right)}_{=p_{(1+\beta)\theta}(x)},
\end{align*}
we obtain 
\[
    r_{\theta}^{(\beta)}
    =
    \frac{1}{1+\beta} \int_{\mathcal{X}}p_{\theta}(x)^{1+\beta} \diff x 
    =
    \frac{1}{1+\beta}
    h^{\beta} \exp \left( v((1+\beta)\theta) 
    -
    (1+\beta) v(\theta) \right).
\]

\section{Proof of Example~\ref{ex:gaussian_mixture}}
\label{app:proof_of_Example_ex:gaussian_mixture}

\subsection{Preparation}

To prove that normal mixture density described in Example~\ref{ex:gaussian_mixture} satisfies the conditions (i)--(iv) in Proposition~\ref{prop:convergence}, we first show the follwing Proposition~\ref{prop:lipschitz_pblp}. 

\begin{prop}
\label{prop:lipschitz_pblp}
Let $\beta>0$ and let $p_{\theta}(x)$ be normal mixture density defined in Example~\ref{ex:gaussian_mixture}. Then, the following hold. 
\begin{enumerate}[{(a)}]
\item $p_{\theta}(x)^{\beta} \frac{\partial}{\partial \theta} \log p_{\theta}(x)$ is Lipschitz, i.e., there exists $L>0$ such that
\begin{align}
\bigg\|
p_{\theta}(x)^{\beta} \frac{\partial}{\partial \theta}\log p_{\theta}(x)
-
p_{\theta'}(x)^{\beta} \frac{\partial}{\partial \theta}\log p_{\theta'}(x)
\bigg\|
\le 
L\|\theta-\theta'\|
\label{eq:lipschitz_pblp}
\end{align}
for all $\theta,\theta' \in \Theta$ and $x \in \mathcal{X}$. 

\item $p_{\theta}(x)$ is Lipschitz, i.e., there exists $L'>0$ such that $|p_{\theta}(x)-p_{\theta'}(x)| \le L'\|\theta-\theta'\|$ for all $\theta,\theta' \in \Theta$ and $x \in \mathcal{X}$.  
\item $\|p_{\theta}(x)^{\beta} \frac{\partial}{\partial \theta} \log p_{\theta}(x)\|$ is integrable, i.e., there exists $B>0$ such that $\int \| p_{\theta}(x)^{\beta} \frac{\partial}{\partial \theta} \log p_{\theta}(x) \| \diff x \le B$ for all $\theta \in \Theta$.

\revisebegin
\item $\{p_{\theta}(x)^{\beta} \frac{\partial}{\partial \theta_k} \log p_{\theta}(x)\}^8$ is integrable for $k=1,2,\ldots,s$, i.e., there exists $C>0$ such that $\int \{p_{\theta}(x)^{\beta} \frac{\partial}{\partial \theta_k} \log p_{\theta}(x) \}^8 \diff x \le C$ for all $\theta \in \Theta$. 
\reviseend
\end{enumerate}
\end{prop}

\begin{proof}[Proposition~\ref{prop:lipschitz_pblp}]
We first prove (a) with the identity 
\begin{align}
&\frac{\partial}{\partial \theta_k}
\left\{
p_{\theta}(x)^{\beta} \frac{\partial}{\partial \theta_j} \log p_{\theta}(x)
\right\} \nonumber \\
&\hspace{5em}=
p_{\theta}(x)^{\beta}
\left[
\beta 
 \frac{\partial}{\partial \theta_j} \log p_{\theta}(x)
 \frac{\partial}{\partial \theta_k} \log p_{\theta}(x)
+
\frac{\partial^2}{\partial \theta_j \partial \theta_k} \log p_{\theta}(x)
\right].
\label{eq:devirative_derivative}
\end{align}

For simplicity, $\phi_1(x),\phi_2(x),\rho$ herein denote $\phi(x;\theta_2,\theta_3^2+\varepsilon),\phi(x;\theta_4,\theta_5^2+\varepsilon)$, and $\rho(\theta_1)$, respectively; then, straightforward calculation leads to
\begin{align*}
\frac{\partial}{\partial \theta_1}\log p_{\theta}(x)
&=
\rho(1-\rho)\{\phi_1(x)-\phi_2(x)\}, \\
\frac{\partial}{\partial \theta_2}\log p_{\theta}(x)
&=
\frac{\rho\phi_1(x)}{p_{\theta}(x)}
\frac{(x-\theta_2)}{\theta_3^2+\varepsilon}, \\
\frac{\partial}{\partial \theta_3}\log p_{\theta}(x)
&=
\frac{\rho\phi_1(x)}{p_{\theta}(x)}
\theta_3
\frac{(x-\theta_2)^2}{(\theta_3^2+\varepsilon)^2}, \\
\frac{\partial^2}{\partial \theta_1^2}\log p_{\theta}(x)
&=
\rho(1-\rho)(1-2\rho)\{\phi_1(x)-\phi_2(x)\}, \\
\frac{\partial^2}{\partial \theta_2^2}\log p_{\theta}(x)
&=
\revisebegin
-\frac{\rho\phi_1(x)}{p_{\theta}(x)}\frac{1}{\theta_3^2+\varepsilon}
+
\frac{\rho \phi_1(x)}{p_{\theta}(x)}
\left\{
1 - \frac{\rho \phi_1(x)}{p_{\theta}(x)} 
\right\}
\left( \frac{x-\theta_2}{\theta_3^2+\varepsilon} \right)^2, \reviseend 
\\
\frac{\partial^2}{\partial \theta_3^2}\log p_{\theta}(x)
&=
-4
\frac{\rho \phi_1(x)}{p_{\theta}(x)}
\theta_3^2 
\frac{(x-\theta_2)^2}{(\theta_3^2+\varepsilon)^3}
+
\frac{\rho \phi_1(x)}{p_{\theta}(x)}
\frac{(x-\theta_2)^2}{(\theta_3^2+\varepsilon)^2} \\
&\hspace{8em}+
\frac{\rho \phi_1(x)}{p_{\theta}(x)}
\left\{
1 - \frac{\rho \phi_1(x)}{p_{\theta}(x)} 
\right\}
\theta_3 
\frac{(x-\theta_2)^4}{(\theta_3^2+\varepsilon)^4}, \\
&\vdots
\end{align*}
Considering for all $x \in \mathbb{R}$ and $\theta \in \Theta$ that 
\begin{itemize} 
\item $\phi_1(x),\phi_2(x) \in [0,1/\sqrt{2\pi \varepsilon}]$,
\item $\rho \phi_1(x)/p_{\theta}(x) \in [0,1]$, 
\item $\rho \in [0,1]$, 
\item $\theta_3/(\theta_3^2+\varepsilon) \in [-1/2\sqrt{\varepsilon},1/2\sqrt{\varepsilon}]$, 
\item and $1/(\theta_3^2+\varepsilon) \in [0,1/\varepsilon]$, 
\end{itemize}
we have 
\begin{align*}
\bigg| \frac{\partial}{\partial \theta_1}
\log p_{\theta}(x) \bigg| 
&\le 
\sqrt{2/\pi \varepsilon} \\
\bigg| \frac{\partial}{\partial \theta_2}
\log p_{\theta}(x) \bigg| 
&\le 
\frac{|x-\theta_2|}{\varepsilon} \\ 
\bigg| \frac{\partial}{\partial \theta_3}
\log p_{\theta}(x) \bigg| 
&\le 
\frac{1}{2\sqrt{\varepsilon}\varepsilon}(x-\theta_2)^2, \\
\bigg| \frac{\partial^2}{\partial \theta_1^2}
\log p_{\theta}(x) \bigg| 
&\le 
\sqrt{2/\pi \varepsilon}, \\
\bigg| \frac{\partial^2}{\partial \theta_2^2}
\log p_{\theta}(x) \bigg| 
&\le 
\revisebegin 
\frac{1}{\varepsilon}
+
\frac{(x-\theta_2)^2}{\varepsilon^2}, \reviseend 
\\
\bigg| \frac{\partial^2}{\partial \theta_3^2}
\log p_{\theta}(x) \bigg| 
&\le 
\frac{2}{\varepsilon^2}(x-\theta_2)^2 
+
\frac{1}{2\sqrt{\varepsilon}\varepsilon^3}(x-\theta_2)^4, \\
&\vdots
\end{align*}
Although we skip for simplicity, by conducting the same redundant calculation for remaining $j,k$, 
the identity \eqref{eq:devirative_derivative} proves the existence of non-negative constants $\{C_l\}_{l=1}^{4}$ and $C_*$ that
\[
    \bigg|
    \frac{\partial}{\partial \theta_k}
    \left\{
        p_{\theta}(x)^{\beta}
        \frac{\partial}{\partial \theta_j}\log p_{\theta}(x)
    \right\}\bigg|
    \le 
    p_{\theta}(x)^{\beta}
    \sum_{l=0}^{4} C_l x^l
    \le C_*
\]
for all $x \in \mathbb{R}$ and $j,k \in \{1,2,\ldots,5\}$. 
Therefore, mean-value theorem finally proves Lipschitz property 
\begin{align*}
\bigg\|
&p_{\theta}(x)^{\beta} \frac{\partial}{\partial \theta}\log p_{\theta}(x)
-
p_{\theta'}(x)^{\beta} \frac{\partial}{\partial \theta}\log p_{\theta'}(x)
\bigg\| \\
&\le 
\left\{
\sum_{j=1}^{5}
\left|
    p_{\theta}(x)^{\beta} \frac{\partial}{\partial \theta_j}\log p_{\theta}(x)
-
p_{\theta'}(x)^{\beta} \frac{\partial}{\partial \theta_j}\log p_{\theta'}(x)
\right|^2
\right\}^{1/2} \\
&\le 
\left\{
\sum_{j=1}^{5}
\sum_{k=1}^{5}
\left[
\sup_{\theta \in \Theta}
\bigg|
\frac{\partial}{\partial \theta_k}
\left\{ p_{\theta}(x)^{\beta}\frac{\partial}{\partial \theta_j}\log p_{\theta}(x) \right\} \bigg|
|\theta_k-\theta_k'|
\right]^2
\right\}^{1/2} \\
&\le \left\{
5C_*^2 \|\theta-\theta'\|^2
\right\}^{1/2}
=:
L\|\theta-\theta'\|^2, 
\quad 
(L:=\sqrt{5}C_*).
\end{align*}

\revisebegin (b)--(d) are also proved in the same way. \reviseend 
\qed
\end{proof}

\subsection{Proof of Example~\ref{ex:gaussian_mixture}}

Leveraging Proposition~\ref{prop:lipschitz_pblp}, we herein prove that the normal mixture density in Example~\ref{ex:gaussian_mixture} satisfies the properties (i)--(iv) in Proposition~\ref{prop:convergence}. 
Particularly, \revisebegin we prove (ii) and (iv) because (i) and (iii) are immediately obtained \reviseend from the differentiability of $p_{\theta}(x)$ and our definition of the stochastic gradient.

\subsubsection*{Proof of (ii)}

We prove the Lipschitz property of the right-hand side of the inequality
\begin{align}
\bigg\|\frac{\partial f(\theta)}{\partial \theta} &- \frac{\partial f(\theta')}{\partial \theta}\bigg\| \nonumber \\
&\le 
\frac{1}{n}\sum_{i=1}^{n}
\bigg\|
    p_{\theta}(x)^{\beta}
    \frac{\partial}{\partial \theta}\log p_{\theta}(x)
    -
    p_{\theta'}(x)^{\beta}
    \frac{\partial}{\partial \theta}\log p_{\theta'}(x)
\bigg\| \nonumber \\
&\hspace{3em}+
\int 
\bigg\|p_{\theta}(x)^{1+\beta}
    \frac{\partial}{\partial \theta}\log p_{\theta}(x)
    -
    p_{\theta'}(x)^{1+\beta}
    \frac{\partial}{\partial \theta}\log p_{\theta'}(x)
\bigg\| \diff x.
\label{eq:f-difference}
\end{align}

The first term of the right-hand side is upper-bounded by $L\|\theta-\theta'\|$ for some $L>0$ as proved in Proposition~\ref{prop:lipschitz_pblp} (a). 
The integrant of the second term is evaluated as
\begin{align*}
\bigg\|p_{\theta}(x)^{1+\beta}&
    \frac{\partial}{\partial \theta}\log p_{\theta}(x)
    -
    p_{\theta'}(x)^{1+\beta}
    \frac{\partial}{\partial \theta}\log p_{\theta'}(x)
\bigg\| \\
&\le 
\bigg\|p_{\theta}(x)
\{p_{\theta}(x)^{\beta}\frac{\partial}{\partial \theta}\log p_{\theta}(x)
-
    p_{\theta'}(x)^{\beta}
    \frac{\partial}{\partial \theta}\log p_{\theta'}(x)\}
\bigg\| \\
&\hspace{10em}+
\bigg\|
\{ p_{\theta}(x) - p_{\theta'}(x) \}
    p_{\theta'}(x)^{\beta}\frac{\partial}{\partial \theta} \log p_{\theta'}(x)
\bigg\| \\
&\le 
p_{\theta}(x) \bigg\| p_{\theta}(x)^{\beta}\frac{\partial}{\partial \theta}\log p_{\theta}(x)
-
    p_{\theta'}(x)^{\beta}
    \frac{\partial}{\partial \theta}\log p_{\theta'}(x) \bigg\| \\
&\hspace{10em}+
|p_{\theta}(x)-p_{\theta'}(x)|
\bigg\| p_{\theta'}(x)^{\beta}\frac{\partial}{\partial \theta} \log p_{\theta'}(x) \bigg\| \\
&\le 
L \|\theta-\theta'\|p_{\theta}(x)
+
L'\|\theta-\theta'\| \bigg\| p_{\theta'}(x)^{\beta}\frac{\partial}{\partial \theta} \log p_{\theta'}(x) \bigg\| \quad (\because \text{Prop.} \ref{prop:lipschitz_pblp}).
\end{align*}
Therefore, the second term in the right-hand side of \eqref{eq:f-difference} is upper-bounded by $\tilde{L}\|\theta-\theta'\|$ with $\tilde{L}:=L+BL'$ by taking the integral with respect to $x$. Therefore, we proved the Lipschitz property of $\partial f(\theta)/\partial \theta$. 
Thus (ii) is proved.

\revisebegin

\subsubsection*{Proof of (iv)}

With 
\begin{align*}
\frac{\partial}{\partial \theta}
r_{\theta}^{\beta}
&=
\frac{\partial}{\partial \theta} 
\left( \frac{1}{1+\beta}\int_{\mathcal{X}} p_{\theta}(x)^{1+\beta} \diff x \right)
=
\int_{\mathcal{X}} p_{\theta}(x)^{1+\beta} \frac{\partial}{\partial \theta} \log p_{\theta}(x) \diff x, 
\end{align*}

we have 
\begin{align*}
\mathbb{E}_{\zeta^{(t)}}&
\left(
\bigg\| g(\theta^{(t)};\zeta^{(t)}) - \frac{\partial f(\theta^{(t)})}{\partial \theta} \bigg\|^2
\right) \\
&=
\mathbb{E}_{\zeta^{(t)}}
\left(
    \bigg\|
    \frac{1}{m}\sum_{j=1}^{m} w_t(y_j^{(t)}) p_{\theta^{(t)}}(y_j^{(t)})^{\beta} \frac{\partial \log p_{\theta^{(t)}}(y_j^{(t)})}{\partial \theta} 
    -
    \frac{\partial}{\partial \theta} r_{\theta^{(t)}}^{\beta}
    \bigg\|^2
\right) \\
&=
\mathbb{E}_{\zeta^{(t)}}
\left(
    \sum_{k=1}^{s}
    \bigg\{
    \frac{1}{m}\sum_{j=1}^{m} w_t(y_j^{(t)}) p_{\theta^{(t)}}(y_j^{(t)})^{\beta} \frac{\partial \log p_{\theta^{(t)}}(y_j^{(t)})}{\partial \theta_k} 
    -
    \frac{\partial}{\partial \theta_k} r_{\theta^{(t)}}^{\beta}
    \bigg\}^2
\right) \\
&=
\sum_{k=1}^{s}
\mathbb{E}_{\zeta^{(t)}}
\left(
    \bigg\{
    \frac{1}{m}\sum_{j=1}^{m} w_t(y_j^{(t)}) p_{\theta^{(t)}}(y_j^{(t)})^{\beta} \frac{\partial \log p_{\theta^{(t)}}(y_j^{(t)})}{\partial \theta_k} 
    -
    \frac{\partial}{\partial \theta_k} r_{\theta^{(t)}}^{\beta}
    \bigg\}^2 \right) \\
&=
\sum_{k=1}^{s}
\mathbb{V}_{\zeta^{(t)}}
\left(
    \frac{1}{m}\sum_{j=1}^{m} w_t(y_j^{(t)}) p_{\theta^{(t)}}(y_j^{(t)})^{\beta} \frac{\partial \log p_{\theta^{(t)}}(y_j^{(t)})}{\partial \theta_k} 
\right) \\
&=
\frac{1}{m}
\sum_{k=1}^{s}
\mathbb{V}_{\zeta^{(t)}}
\left(
    w_t(y_1^{(t)}) p_{\theta^{(t)}}(y_1^{(t)})^{\beta} \frac{\partial \log p_{\theta^{(t)}}(y_1^{(t)})}{\partial \theta_k} 
\right) \\
&\le 
\frac{1}{m}
\sum_{k=1}^{s}
\mathbb{E}_{\zeta^{(t)}}
\left(
\left\{
    w_t(y_1^{(t)}) p_{\theta^{(t)}}(y_1^{(t)})^{\beta} \frac{\partial \log p_{\theta^{(t)}}(y_1^{(t)})}{\partial \theta_k} 
\right\}^2
\right) \\
&\le 
\frac{1}{m}
\sum_{k=1}^{s}
\underbrace{
\mathbb{E}_{\zeta^{(t)}}\left( w_t(y_1^{(t)})^4 \right) ^{1/2}
}_{(\star 1)}
\underbrace{
\mathbb{E}_{\zeta^{(t)}}\left(
\left\{
p_{\theta^{(t)}}(y_1^{(t)})^{\beta} \frac{\partial \log p_{\theta^{(t)}} (y_1^{(t)})}{\partial \theta_k} 
\right\}^4
\right)^{1/2}
}_{(\star 2)},
\end{align*}
where the last inequality is known as Cauchy–Schwarz inequality. 
The boundedness of ($\star 1$) immediately follows from the assumption (C-3) in Example 1, and the boundedness of ($\star 2$) is obtained by again applying Cauchy-Schwarz inequality:
\begin{align*}
    (\star 2)
    &=
    \left(
        \int  
        \left\{ p_{\theta^{(t)}}(y)^{\beta} \frac{\log p_{\theta^{(t)}}(y)}{\partial \theta_k} \right\}^4 
        \tilde{p}_t(y)
        \diff y
    \right)^{1/2} \\
    &\le 
    \left(
        \int  
        \left\{ p_{\theta^{(t)}}(y)^{\beta} \frac{\log p_{\theta^{(t)}}(y)}{\partial \theta_k} \right\}^8
        \diff y
    \right)^{1/4}
    \left(
        \int  
        \tilde{p}_t(y)^2
        \diff y
    \right)^{1/4}.
\end{align*}
Boundedness of the first term is proved by Proposition~\ref{prop:lipschitz_pblp} (d). 
The second term is also bounded by the squared-integrable assumption on $\tilde{p}_t$. 
Thus (iv) is proved. 

\reviseend

\qed

\section{Distributions} 
\label{app:distribution}

\begin{enumerate}[{(1)}]
\item \textbf{Inverse normal distribution:} 
the probability density function is
\[
    p^{\ig}_{\theta}(x) 
    =
    \sqrt{\frac{\lambda}{2\pi x^3}}
    \exp\left(
        -\frac{\lambda (x-\mu)^2}{2\mu^2 x}
    \right), 
    \quad 
    (x>0)
\]
where $\mu>0$ is the mean parameter and $\lambda>0$ is the shape parameter. 
As we have 
\[
    t_{\theta}^{\ig}(x)
    =
    \left(
        \frac{\lambda (x-\mu)}{\mu^3} \, , \,
        \frac{1}{2\lambda} - \frac{(x-\mu)^2}{2\mu^2 x}
    \right),
    \quad 
    \theta=(\mu,\lambda),
\]
the maximum likelihood estimator for $\theta$ is 
\[
    \hat{\mu}=\frac{1}{n}\sum_{i=1}^{n} x_i,
    \quad 
    \hat{\lambda}=\frac{1}{n^{-1} \sum_{i=1}^{n} \{x_i^{-1} - \hat{\mu}^{-1} \}}.
\]
We leverage \verb|rinvgauss| function in \verb|actuar| package in \verb|R| language to generate random nubmers following the inverse normal distribution.

\item \textbf{Gompertz distribution:} 
the probability density function is
\[
    p^{\gompertz}_{\theta}(x)
    =
    \lambda \exp\left(\omega x + \frac{\lambda}{\omega} \{1-\exp(\omega x)\}\right), \quad (x \ge 0)
\]
where $\omega>0$ is the scale parameter and $\lambda>0$ is the shape parameter. 
As we have 
\[
    t_{\theta}^{\gompertz}(x)
    =
    \left(
        x - 
        \lambda \left( \frac{1-\exp(\omega x)}{\omega^2} + \frac{x \exp(\omega x)}{\omega} \right) 
        \, , \, 
        \frac{1}{\lambda} + \frac{1-\exp(\omega x)}{\omega}
    \right),
    \quad 
    \theta=(\omega,\lambda),
\]
the maximum likelihood estimator satisfies 
\begin{align*}
    \hat{\lambda} &= 
    -\frac{\hat{\omega}}{n^{-1}\sum_{j=1}^{n}\{1-\exp(\hat{\omega} x_j)\}}, \\
\sum_{i=1}^{n}x_i
+&
\frac{1}{n^{-1}\sum_{j=1}^{n}\{1-\exp(\hat{\omega}x_j)\}}
\sum_{i=1}^{n}
\left\{ \frac{1-\exp(\hat{\omega} x_i)}{\hat{\omega}} 
+ 
x_i \exp(\hat{\omega} x_i) \right\} 
=
0.
\end{align*}
We can numerically find $\hat{\omega}$ by the Newton-Raphson algorithm, whereby we obtain $\hat{\lambda}$. 
We leverage \verb|rgompertz| function in \verb|VGAM| package in \verb|R| language to generate random numbers following the gompertz distribution.

\item \textbf{Normal mixture distribution:} 
the probability density function is
\[
    p_{\theta}^{\gm}(x)
    =
    \alpha \phi(x;\mu_1,\sigma_1^2)
    +
    (1-\alpha) \phi(x;\mu_2,\sigma_2^2),
    \quad 
    \phi(x;\mu,\sigma)=\frac{1}{\sqrt{2\pi\sigma^2}} \exp\left(-\frac{(x-\mu)^2}{2\sigma^2}\right),
\]
where $\alpha \in [0,1]$ denotes the mixing coefficient. 
We have 
\begin{align*}    
t_{\theta}^{\gm}(x)
    &=
    \left(
        c_1 \frac{x-\mu_1}{\sigma_1^2},
        c_1 \left\{
            \frac{(x-\mu_1)^2}{\sigma_1^3} - \frac{1}{\sigma_1}
        \right\},
        c_2 \frac{x-\mu_2}{\sigma_2^2},
        c_2 \left\{
            \frac{(x-\mu_2)^2}{\sigma_2^3} - \frac{1}{\sigma_2}
        \right\},
        c_3
    \right), \\
    \theta&=(\mu_1,\sigma_1,\mu_2,\sigma_2,\alpha),
\end{align*}
where
\[
    c_1=\frac{\alpha \phi(x;\mu_1,\sigma_1^2)}{p_{\theta}^{\gm}(x)}, \quad
    c_2=\frac{(1-\alpha) \phi(x;\mu_2,\sigma_2^2)}{p_{\theta}^{\gm}(x)}, \quad 
    c_3=\frac{\phi(x;\mu_1,\sigma_1^2)-\phi(x;\mu_2,\sigma_2^2)}{p_{\theta}^{\gm}(x)}.
\]
We computed the maximum likelihood estimator by leveraging the \verb|GMM| fucntion in \verb|ClusterR| package in \verb|R| language. 
\end{enumerate}

\bibliographystyle{apalike}
\bibliography{cgd}

\end{document}